\documentclass[11pt]{llncs}

\usepackage{amssymb}
\usepackage{amsmath}
\usepackage{epsfig}
\usepackage{graphics}
\usepackage{color}
\usepackage{llncsdoc}
\usepackage{fullpage}
\usepackage{xspace}

\def\cF{\mathcal{F}}
\def\cE{\mathcal{E}}
\def\cS{\mathcal{S}}
\def\N{{\mathbb N}}
\def\NP{\textrm{NP}}
\def\coNP{\textrm{coNP}}
\def\poly{\textrm{poly}}
\def\col{\textrm{col}}

\def\CPMlong{\textsc{Multicolored Perfect Matching}\xspace}
\def\CPM{\textsc{MPM}\xspace}

\def\TDM{$3$-\textsc{Dimensional Perfect Matching}\xspace}
\def\dSC{\textsc{$d$-Set Cover}\xspace}
\def\tESC{\textsc{$3$-Exact Set Cover}\xspace}
\def\IS{\textsc{Independent Set}\xspace}
\def\IMlong{\textsc{Induced Matching}\xspace}
\def\IM{\textsc{IM}\xspace}

\def\RBDSlong{\textsc{Red Blue Dominating Set}\xspace}
\def\RBDS{\textsc{RBDS}\xspace}
\def\DSlong{\textsc{Dominating Set}\xspace}
\def\DS{\textsc{DS}\xspace}
\def\IDSlong{\textsc{Independent Dominating Set}\xspace}
\def\IDS{\textsc{IDS}\xspace}
\def\ConDS{\textsc{Connected Dominating Set}\xspace}
\def\MC{\textsc{Multicolored Clique}\xspace}
\def\ConVClong{\textsc{Connected Vertex Cover}\xspace}
\def\ConVC{\textsc{ConVC}\xspace}
\def\CapVClong{\textsc{Capacitated Vertex Cover}\xspace}
\def\CapVC{\textsc{CapVC}\xspace}
\def\ConFVS{\textsc{Connected Feedback Vertex Set}\xspace}

\spnewtheorem{reduction-rule}[rule]{Rule}{\bfseries}{\itshape}

\begin{document}

\title{Tight Kernel Bounds for Problems on\\ Graphs with Small Degeneracy\thanks{Partially supported by ERC Starting Grant NEWNET 279352.}}
\author{Marek Cygan\inst{1} \and Fabrizio Grandoni\inst{1} \and Danny Hermelin\inst{2}}

\institute{%
Dalle Molle Institute (IDSIA)\\
Galleria 1, 6928, Manno - Switzerland\\
%\email{marek@idsia.ch, fabrizio@idsia.ch}%
\email{\{marek,fabrizio\}@idsia.ch}%
\and%
Max Plank Institute for Informatics, \\
Saarbr\"{u}cken, Saarland 66111 - Germany \\
\email{hermelin@mpi-inf.mpg.de}%
}

\date{}

\maketitle

%%%%%%%%%%%%%%%%%%%%%%%%%%%%%%%%%%%%%%%%%%%%%%%%%%%%%%%%%%%%%%%%%%%%%
%%%%%%%%% Abstract
%%%%%%%%%%%%%%%%%%%%%%%%%%%%%%%%%%%%%%%%%%%%%%%%%%%%%%%%%%%%%%%%%%%%%

\begin{abstract}

Kernelization is a strong and widely-applied technique in
parameterized complexity. In a nutshell, a kernelization
algorithm for a parameterized problem transforms a given
instance of the problem into an equivalent instance whose size
depends solely on the parameter. Recent years have seen major
advances in the study of both upper and lower bound techniques
for kernelization, and by now this area has become one of the
major research threads in parameterized complexity.

\hspace{10pt} In this paper we consider kernelization for
problems on $d$-degenerate graphs, i.e. graphs such that any
subgraph contains a vertex of degree at most $d$. This graph
class generalizes many classes of graphs for which effective
kernelization is known to exist, e.g. planar graphs, $H$-minor
free graphs, and $H$-topological-minor free graphs. We show
that for several natural problems on $d$-degenerate graphs the
best known kernelization upper bounds are essentially tight. In
particular, using intricate constructions of weak compositions,
we prove that unless \coNP $\subseteq$ \NP/\poly:
\begin{itemize}
\item \DSlong has no kernels of size
    $O(k^{(d-1)(d-3)-\varepsilon})$ for any $\varepsilon >
    0$. The current best upper bound is $O(k^{(d+1)^2})$.
\item \IDSlong has no kernels of size $O(k^{d-4-\varepsilon})$
    for any $\varepsilon > 0$. The current best upper bound
    is $O(k^{d+1})$.
\item \IMlong has no kernels of size $O(k^{d-3-\varepsilon})$
    for any $\varepsilon > 0$. The current best upper bound
    is $O(k^d)$.
\end{itemize}
To the best of our knowledge, the result on \DSlong is the
first example of a lower bound with a super-linear dependence
on $d$ in the exponent.

\hspace{10pt} In the last section  of the paper, we also give simple
kernels for \ConVClong and \CapVClong of size $O(k^d)$ and
$O(k^{d+1})$ respectively. We show that the latter problem has
no kernels of size $O(k^{d-\varepsilon})$ unless \coNP
$\subseteq$ \NP/\poly\ by a simple reduction from \dSC (a
similar lower bound for \ConVClong is already known).
\end{abstract}

%%%%%%%%%%%%%%%%%%%%%%%%%%%%%%%%%%%%%%%%%%%%%%%%%%%%%%%%%%%%%%%%%%%%%
%%%%%%%%% Section: Introduction
%%%%%%%%%%%%%%%%%%%%%%%%%%%%%%%%%%%%%%%%%%%%%%%%%%%%%%%%%%%%%%%%%%%%%

\section{Introduction}
\label{Section: Introduction}

Parameterized complexity is a two-dimensional refinement of
classical complexity theory introduced by Downey and
Fellows~\cite{DowneyFellows1999} where one takes into account
not only the total input length $n$, but also other aspects of
the problem quantified in a numerical parameter $k \in \N$. The
main goal of the field is to determine which problems have
algorithms whose exponential running time is confined strictly
to the parameter. In this way, algorithms running in $f(k)
\cdot n^{O(1)}$ time for some computable function $f()$ are
considered feasible, and parameterized problems that admit
feasible algorithms are said to be \emph{fixed-parameter
tractable}. This notion has proven extremely useful in
identifying tractable instances for generally hard problems,
and in explaining why some theoretically hard problems are
solved routinely in practice.

A closely related notion to fixed-parameter tractability is
that of kernelization. A \emph{kernelization algorithm} (or
\emph{kernel}) for a parameterized problem $L \subseteq
\{0,1\}^* \times \N$ is a polynomial time algorithm that
transforms a given instance $(x,k)$ to an instance $(x',k')$
such that: $(i)$ $(x,k) \in L \iff (x',k') \in L$, and $(ii)$
$|x'| + k' \leq f(k)$ for some computable function $f$. In
other words, a kernelization algorithm is a polynomial-time
reduction from a problem to itself that shrinks the problem
instance to an instance with size depending only on the
parameter. Appropriately, the function $f$ above is called the
\emph{size} of the kernel.

Kernelization is a notion that was developed in parameterized
complexity, but it is also useful in other areas of computer
science such as cryptography~\cite{HarnikNaor2010} and
approximation algorithms~\cite{NemhauserTrotter1975}. In
parameterized complexity, not only is it one of the most
successful techniques for showing positive results, it also
provides an equivalent way of defining fixed-parameter
tractability: A decidable parameterized problem is solvable in
$f(k) \cdot n^{O(1)}$ time iff it has a
kernel~\cite{CaiChenDowneyFellows1997}. From practical point of
view, compression algorithms often lead to efficient
preprocessing rules which can significantly simplify real life
instances~\cite{Fellows2006,GuoNiedermeier2007b}. For these
reasons, the study of kernelization is one of the leading
research frontiers in parameterized complexity. This research
endeavor has been fueled by recent tools for showing lower
bounds on kernel
sizes~\cite{Bodlaender-et-al2009a,Bodlaender-et-al2011,Bodlaender-et-al2009b,ChenFlumMuller2011,DellMarx2012,DellMelkebeek2010,DomLokshtanovSaurabh2009,HermelinWu2012,Kratsch2012}
which rely on the standard complexity-theoretic assumption of
\coNP $\nsubseteq$ \NP/\poly.

Since a parameterized problem is fixed-parameter tractable iff
it is kernelizable, it is natural to ask which fixed-parameter
problems admit kernels of reasonably small size. In recent
years there has been significant advances in this area. One
particulary prominent line of research in this context is the
development of \emph{meta-kernelization} algorithms for
problems on sparse graphs. Such algorithms typically provide
compressions of either linear or quadratic size to a wide range
of problems at once, by identifying certain generic problem
properties that allow for good compressions. The first work in
this line of research is due to Guo and
Niedermeier~\cite{GuoNiedermeier2007a}, which extended the
ideas used in the classical linear kernel for \DS\ in planar
graphs~\cite{AlberFellowsNiedermeier2004} to linear kernels for
several other planar graph problems. This result was subsumed
by the seminal paper of Bodlaender \emph{et
al.}~\cite{Bodlaender-et-al2009c}, which provided
meta-kernelization algorithms for problems on graphs of bounded
genus, a generalization of planar graphs. Later Fomin \emph{et
al.}~\cite{Fomin-et-al2010} provided a meta-kernel for problems
on $H$-minor free graphs which include all bounded genus
graphs. Finally, a recent manuscript by Langer \emph{et
al.}~\cite{Langer-et-al2012} provides a meta-kernelization
algorithm for problems on $H$-topological-minor free graphs.
All meta-kernelizations above have either linear or quadratic
size.

How far can these meta-kernelization results be extended? A
natural class of sparse graphs which generalizes all graph
classes handled by the meta-kernelizations discussed above is
the class of $d$-degenerate graphs. A graph is called
\emph{$d$-degenerate} if each of its subgraphs has a vertex of
degree at most $d$. This is equivalent to requiring that the
vertices of the graph can be linearly ordered such that each
vertex has at most $d$ neighbors to its right in this ordering.
For example, any planar graph is 5-degenerate, and for any
$H$-minor (resp. $H$-topological-minor) free graph class there
exists a constant $d(H)$ such that all graphs in this class are
$d(H)$-degenerate (see \emph{e.g.}~\cite{Diestel2005}). Note
that the \IS\ problem has a trivial linear kernel in
$d$-degenerate graphs, which gives some hope that a
meta-kernelization result yielding small degree polynomial
kernels might be attainable for this graph class.

Arguably the most important kernelization result in
$d$-degenerate graphs is due to Philip \emph{et
al.}~\cite{Philip-et-al2012} who showed a $O(k^{(d+1)^2})$ size
kernel for \DSlong, and an $O(k^{d+1})$ size kernel for
\IDSlong. Erman \emph{et al.}~\cite{Erman-et-al2010} and Kanj
\emph{et al.}~\cite{Kanj-et-al2011} independently gave a
$O(k^d)$ kernel for the \IMlong\ problem, while Cygan \emph{et
al.}~\cite{Cygan-et-al2012} showed a $O(k^{d+1})$ kernel is for
\ConVClong. While all these results give polynomial kernels,
the exponent of the polynomial depends on $d$, leaving open the
question of kernels of polynomial size with a fixed constant
degree. This question was answered negatively for \ConVC\
in~\cite{Cygan-et-al2012} using the standard reduction from
\dSC. It is also shown in~\cite{Cygan-et-al2012} that other
problems such \ConDS\ and \ConFVS\ do not admit a kernel of any
polynomial size unless \coNP $\subseteq$ \NP/\poly.
Furthermore, the results in~\cite{DellMarx2012,HermelinWu2012}
can be easily used to show exclude a
$O(k^{d-\varepsilon})$-size kernel for \DSlong, for some small
positive constant $\varepsilon$.\\

\noindent \textbf{Our results:} In this paper, we show that all
remaining kernelization upper bounds for $d$-degenerate graphs
mentioned above have matching lower bounds up to some small
additive constant. Perhaps the most surprising result we obtain
is the exclusion of $O(k^{(d-3)(d-1)-\varepsilon})$ size
kernels for \DSlong\ for any $\varepsilon > 0$, under the
assumption of coNP $\nsubseteq$ \NP/\poly. This result is
obtained by an intricate application of \emph{weak
compositions} which were introduced
by~\cite{DellMelkebeek2010}, and further applied
in~\cite{DellMarx2012,HermelinWu2012}. What makes this result
surprising is that it implies that \IDSlong\ is fundamentally
easier than \DSlong\ in $d$-degenerate graphs. We also show a
$O(k^{d-4-\varepsilon})$ lower bound for \IDSlong, and an
$O(k^{d-3-\varepsilon})$ lower bound for \IMlong. The latter
result is also somewhat surprising when one considers the
trivial linear kernel for the closely related \IS\ problem.
Finally, we slightly improve the $O(k^{d+1})$ kernel for
\ConVClong\ of~\cite{Cygan-et-al2012} to $O(k^d)$, and show
that the related \CapVClong\ problem has a kernel of size
$O(k^{d+1})$, but no kernel of size $O(k^{d-\varepsilon})$
unless coNP $\subseteq$ \NP/\poly. Table~\ref{Table} summarizes
the currently known state of the art of kernel sizes for the
problems considered in this paper.

\begin{table}%[b]
\begin{center}
\begin{tabular}{|c|c|c|}
\hline
& \quad Lower Bound \quad & \quad Upper Bound \quad \\ \hline
\DSlong & $(d-3)(d-1)-\varepsilon$ & $(d+1)^2$~\cite{Philip-et-al2012} \\ \hline
\quad \IDSlong \quad & $d-4-\varepsilon$ & $d+1$~\cite{Philip-et-al2012} \\ \hline
\IMlong & $d-3-\varepsilon$ & $d$~\cite{Erman-et-al2010,Kanj-et-al2011} \\ \hline
\ConVClong & $d-1-\varepsilon$~\cite{Cygan-et-al2012} & $d$ \\ \hline
\CapVClong & $d-\varepsilon$ & $d+1$ \\ \hline
%\IEDS\ & &\\ \hline
\end{tabular}
\end{center}
\caption{Lower and upper bounds for kernel sizes for problems
in $d$-degenerate graphs. Only the exponent of the polynomial
in $k$ is given. Results without a citation are obtained in
this paper.}
\label{Table}%
\end{table}

%%%%%%%%%%%%%%%%%%%%%%%%%%%%%%%%%%%%%%%%%%%%%%%%%%%%%%%%%%%%%%%%%%%%%
%%%%%%%%% Section: Kernelization Lower Bounds
%%%%%%%%%%%%%%%%%%%%%%%%%%%%%%%%%%%%%%%%%%%%%%%%%%%%%%%%%%%%%%%%%%%%%

\section{Kernelization Lower Bounds}
\label{Section: Kernelization Lower Bounds}

In the following section we quickly review the main tool that
we will be using for showing our kernelization lower bounds,
namely compositions. A composition algorithm is typically a
transformation from a classical NP-hard problem $L_1$ to a
parameterized problem $L_2$. It takes as input a sequence of
$T$ instances of $L_1$, each of size $n$, and outputs in
polynomial time an instance of $L_2$ such that $(i)$ the output
is a YES-instance iff one of the inputs is a YES-instance, and
$(ii)$ the parameter of the output is polynomially bounded by
$n$ and has only ``small" dependency on $T$. Thus, a
composition may intuitively be thought of as an ``OR-gate" with
a guarantee bound on the parameter of the output. More
formally, for an integer $d \geq 1$, a weak $d$-composition is
defined as follows:

\begin{definition}[weak $d$-composition]
\label{Definition: d-Composition}%
Let $d \geq 1$ be an integer constant, let $L_1 \subseteq
\{0,1\}^*$ be a classical (non-parameterized) problem, and let
$L_2 \subseteq \{0,1\}^* \times \N$ be a parameterized problem.
A \emph{weak $d$-composition} from $L_1$ to $L_2$ is a
polynomial time algorithm that on input $x_1, \dots, x_{t^d}
\in \{0,1\}^n$ outputs an instance $(y, k') \in \{0,1\}^*
\times \N$ such that:
\begin{itemize}
\item $(y, k') \in L_2 \Longleftrightarrow x_i \in
    L_1$ for some $i$, and
\item $k' \leq t \cdot n^{O(1)}$.
\end{itemize}
\end{definition}

The connection between compositions and kernelization lower
bounds was discovered by~\cite{Bodlaender-et-al2009a} using
ideas from~\cite{HarnikNaor2010} and a complexity theoretic
lemma of~\cite{FortnowSanthanam2008}. The following particular
connection was first observed in~\cite{DellMelkebeek2010}.

\begin{lemma}[\cite{DellMelkebeek2010}]
\label{Lemma: d-composition gives lower bounds}%
Let $d \geq 1$ be an integer, let $L_1 \subseteq \{0,1\}^*$ be
a classical \textnormal{NP}-hard problem, and let $L_2
\subseteq \{0,1\}^* \times \N$ be a parameterized problem. A
weak-$d$-composition from $L_1$ to $L_2$ implies that $L_2$ has
no kernel of size $O(k^{d-\varepsilon})$ for any $\varepsilon
> 0$, unless \textnormal{coNP} $\subseteq$
\textnormal{NP/poly}.
\end{lemma}

\begin{remark}
\label{Remark: Compression}%
Lemma~\ref{Lemma: d-composition gives lower bounds} also holds
for \emph{compressions}, a stronger notion of kernelization, in
which the reduction is not necessarily from the problem to
itself, but rather from the problem to some arbitrary~set.
\end{remark}

%%%%%%%%%%%%%%%%%%%%%%%%%%%%%%%%%%%%%%%%%%%%%%%%%%%%%%%%%%%%%%%%%%%%%
%%%%%%%%% Section: Dominating Set
%%%%%%%%%%%%%%%%%%%%%%%%%%%%%%%%%%%%%%%%%%%%%%%%%%%%%%%%%%%%%%%%%%%%%

\section{Dominating Set}
\label{Section: Dominating Set}

\def\iteri{{\delta}}
\def\iterj{{\lambda}}
\def\iterc{{\gamma}}

We begin by considering the \DSlong (\DS) problem. In this
problem, we are given an undirected graph $G=(V,E)$ together
with an integer $k$, and we are asked whether there exists a
set $S$ of at most $k$ vertices such that each vertex of $G$
either belongs to $S$ or has a neighbor in $S$ (i.e. $N[S]=V$).
The main result of this section is stated in
Theorem~\ref{thm:ds} below.

\begin{theorem}
\label{thm:ds} Let $d \geq 4$. The \DSlong problem in
$d$-degenerate graphs has no kernel of size
$O(k^{(d-1)(d-3)-\varepsilon})$ for any constant
$\varepsilon>0$ unless \textnormal{NP} $\subseteq$
\textnormal{coNP/poly}.
\end{theorem}

In order to prove Theorem~\ref{thm:ds}, we show a lower bound
for a similar problem called the \RBDSlong problem (\RBDS):
Given a bipartite graph $G=(R \cup B,E)$ and an integer $k$,
where $R$ is the set of \emph{red} vertices and $B$ is the set
of \emph{blue} vertices, determine whether there exists a set
$D \subseteq R$ of at most $k$ red vertices which dominate all
the blue vertices (i.e. $N(D)=B$). According to
Remark~\ref{Remark: Compression}, the lemma below shows that
focusing on \RBDS is sufficient for proving
Theorem~\ref{thm:ds} above.

\begin{lemma}
\label{lem:rbds} There is a polynomial time algorithm, which
given a $d$-degenerate instance $I=(G=(R \cup B,E),k)$ of \RBDS
creates a $(d+1)$-degenerate instance $I'=(G',k')$ of \DS, such
that $k'=k+1$ and $I$ is YES-instance iff $I'$ is a
YES-instance.
\end{lemma}

\begin{proof}
As the graph $G'$ we initially take $G=(R \cup B,E)$ and then
we add two vertices $r$, $r'$ and make $r$ adjacent to all the
vertices in $R \cup \{r'\}$. Clearly $G'$ is
$(d+1)$-degenerate. Note that if $S \subseteq R$ is a solution
in $I$, then $S \cup \{r\}$ is a dominating set in $I'$. In the
reverse direction, observe that w.l.o.g. we may assume that a
solution $S'$ for $I'$ contains $r$ and moreover contains no
vertex of $B$. Therefore $I$ is a YES-instance iff $I'$ is a
YES-instance. \qed
\end{proof}

We next describe a weak $d(d+2)$-composition from \CPMlong to
\RBDS in $(d+2)$-degenerate graphs. The \CPMlong problem (\CPM)
is as follows: Given an undirected graph $G=(V,E)$ with even
number $n$ of vertices, together with a color function $\col :
E \rightarrow \{0,\ldots,n/2-1\}$, determine whether there
exists a perfect matching in $G$ with all the edges having
distinct colors. A simple reduction from \TDM, which is
NP-complete due to Karp~\cite{karp-21}, where we encode one
coordinate using colors, shows that \CPM is NP-complete
when we consider multigraphs. In the following lemma
we show that \CPM is NP-complete even for simple graphs.

\begin{lemma}
The \CPM problem is NP-complete.
\end{lemma}

\begin{proof}
We show a polynomial time reduction from \CPM in multigraphs,
where several parallel edges are allowed.
Let $(G=(V,E),\col : V \to \{0, \ldots, |V|/2-1\})$ 
be an instance of \CPM where $G$ is a multigraph, $|V|$ is even
and $\pi : E \to \{0, \ldots, |E|-1\}$ is an arbitrary
bijection between the multiset of edges and integers from $0$ to $|E|-1$.
Moreover we assume that there is some fixed linear order on $V$.
We create a graph $G'$ as follows.
The set of vertices of $G'$ is $V' = \{x_v : v \in V\} \cup \{y_{v,e} : e \in E, v \in V\}$, that is we create a vertex in $G'$ for each vertex of $G$
and for each endpoint of an edge of $G$.
The set of edges of $G'$ is $E' = \{y_{u,e}y_{v,e} : e=uv \in E\} \cup \{x_v y_{v,e} : v \in V, e \in E, v \in e\}$.
Finally, we define the color function $\col': E' \to \{0, \ldots, n'/2-1\}$,
where $n' = 2|E| + |V|$.
For $y_{u,e}y_{v,e} \in E'$ we set $\col'(y_{u,e}y_{v,e}) = \pi(e)$,
for $x_v y_{v,e} \in E'$, where $e=uv$ and $u<v$ we set $\col'(x_v y_{v,e}) = \pi(e)$,
while for $x_v y_{v,e} \in E'$, where $e=uv$ and $u>v$ we set $\col'(x_v y_{v,e}) = |E|+\col(e)$.

Clearly the construction can be performed in polynomial time,
and the graph $G'$ is simple, hence it suffices to show that the instance
$(G',\col')$ is a YES-instance iff $(G,\col)$ is a YES-instance.
Let $M \subseteq E$ be a solution for $(G,\col)$.
Observe that the set $M' = \{y_{u,e}y_{v,e} : e=uv \in M\} \cup \{x_{v,e}y_{v,e} : e \in E \setminus M, v \in e\}$ is a solution for $(G',\col')$.

In the other direction, assume that $M' \subseteq E'$ is a solution for $(G',\col')$.
Note that for each $e=uv \in E$ either we have $y_{u,e}y_{v,e} \in M'$
or both $x_uy_{u,e}$ and $x_vy_{v,e}$ belong to $M'$.
Consequently we define $M \subseteq E$ to
be the set of edges $e=uv$ of $E$ such that $y_{u,e}y_{v,e} \not\in M'$.
It is easy to verify that $M$ is a solution for $(G,\col)$.
\end{proof}

The construction of the weak composition is rather involved. We
construct an instance graph $H_{inst}$ which maps feasible
solutions of each \CPM instance into feasible solutions of the
\RBDS instance. Then we add an enforcement gadget
$(H_{enf},E_{conn})$ which prevents partial solutions of two or
more \CPM instances to form altogether a solution for the \RBDS
instance. The overall \RBDS instance will be denoted by
$(H,k)$, where $H$ is the union of $H_{inst}$ and $H_{enf}$
along with the edges $E_{conn}$ that connect between these
graphs. The construction of the instance graph is relatively
simple, while the enforcement gadget is rather complex. In the
next subsection we describe $H_{enf}$ and its crucial
properties. In the following subsection we describe the rest of
the construction, and prove the claimed lower bound on \RBDS
(and hence \DS). Both $H_{enf}$ and $H_{inst}$ contain red and
blue nodes. We will use the convention that $R$ and $B$ denote
sets of red and blue nodes, respectively. We will use $r$ and
$b$ to indicate red and blue nodes, respectively. A color is
indicated by $\ell$.

\subsection{The Enforcement Graph}

The enforcement graph $H_{enf}=(R_{enf}\cup B_{enf},E_{enf})$
is a combination of $3$ different gadgets: the \emph{encoding
gadget}, the \emph{choice gadget}, and the \emph{fillin gadget}
(see also Fig.~\ref{fig:ds}), i.e. $R_{enf} = R_{code} \cup
R_{fill}$ and $B_{enf} = B_{code} \cup B_{choice} \cup
B_{fill}$ ($R_{choice}$ is empty).\\

%\paragraph*{\bf Encoding gadget}

\noindent \textbf{Encoding gadget:} The role of this gadget is
to encode the indices of all the instances by different partial
solutions. It consists of nodes $R_{code}\cup B_{code}$, plus
the edges among them. The set $R_{code}$ contains one node
$r_{\delta,\lambda,\gamma}$ for all integers $0\leq
\delta<d+2$, $0\leq \lambda<d$, and $0\leq \gamma<t$. In
particular, $|R_{code}|=(d+2)dt$. The set $B_{code}$ is the
union of sets $B^{\ell}_{code}$ for each color $0\leq
\ell<n/2$. In turn, $B^{\ell}_{code}$ contains a node
$b^{\ell}_a$ for each integer $0\leq a<(dt)^{d+2}$. We connect
nodes $r_{\delta,\lambda,\gamma}$ and $b^{\ell}_a$ iff
$a_\delta=\lambda \cdot t + \gamma$, where
$(a_0,\ldots,a_{d+1})$ is the expansion of $a$ in base $dt$,
i.e. $a=\sum_{0 \le \delta <d+2} a_\delta(dt)^\delta$. There is
a subtle reason behind this connection scheme, which hopefully
will be clearer soon. Note that since $0 \le \gamma < t$, pairs
$(\lambda,\gamma)$ are in one to one correspondence with
possible values of digits $a_\delta$.\\

%\paragraph*{\bf Choice gadget}

\noindent \textbf{Choice gadget:} The role of the choice gadget
is to  guarantee the following \emph{choice property}: Any
feasible solution to the overall \RBDS instance $(H,k)$
contains all nodes $R_{code}$ except possibly one node
$r_{\delta,\lambda,\gamma_{\delta,\lambda}}$ for each pair
$(\delta,\lambda)$ (hence at least $(d+2)d(t-1)$ nodes of
$R_{code}$ altogether are taken). Intuitively, the
$\gamma_{\delta,\lambda}$'s will be used to identify the index
of one \CPM input instance. In order to do that, we introduce a
set of nodes $B_{choice}$, containing a node
$b_{\delta,\lambda,\gamma_1,\gamma_2}$ for every pair
$(\delta,\lambda)$ and for every $0\leq \gamma_1<\gamma_2<t$.
We connect $b_{\delta,\lambda,\gamma_1,\gamma_2}$ to both
$r_{\delta,\lambda,\gamma_1}$ and
$r_{\delta,\lambda,\gamma_2}$. It is not hard to see that, in
order to dominate $B_{choice}$, it is necessary and sufficient
to select from $R_{code}$ a subset of nodes with the choice
property.\\

%\paragraph*{\bf Fillin gadget}

\noindent \textbf{Fillin gadget:} We will guarantee that, in
any feasible solution, precisely $(d+2)d(t-1)$ nodes from
$R_{code}$ are selected. Given that, for each pair
$(\delta,\lambda)$, there will be precisely one node
$r_{\delta,\lambda,\gamma_{\delta,\lambda}}$  which is not
included in the solution. Consequently, as we will prove, for
each $0 \le \ell < n/2$ in $B^{\ell}_{code}$ there will be
exactly $d^{d+2}$ uncovered nodes, namely the nodes
$b^{\ell}_{a}=b^{\ell}_{(a_0,\ldots,a_{d+1})}$ such that for
each $0\leq \delta< d+2$ and $\lambda t\leq a_{\delta}<
(\lambda +1)t$ one has $a_{\delta}=\lambda t +
\gamma_{\delta,\lambda}$. Ideally, we would like to cover such
nodes by means of red nodes in the instance graph $H_{inst}$
(to be defined later), which encode a feasible solution to some
\CPM instance. However, the degeneracy of the overall graph
would be too large. The role of the fillin gadget is to
circumvent this problem, by leaving at most $d$ uncovered nodes
in each $B^{\ell}_{code}$. The fillin gadget consists of nodes
$R_{fill}\cup B_{fill}$, with some edges incident to them. The
set $R_{fill}$ is the union of sets $R^{\ell}_{fill}$ for each
color $\ell$. In turn $R^{\ell}_{fill}$ contains one node
$r^{\ell}_{a,j}$ for each $1\leq j\leq d^{d+2}-d$ and $0\leq
a<(dt)^{d+2}$. We connect each $r^{\ell}_{a,j}$ to
$b^{\ell}_{a}$. The set $B_{fill}$ contains one node
$b^{\ell}_j$, for each color $\ell$ and for all $1\leq j\leq
d^{d+2}-d$. We connect $b^{\ell}_j$ to all nodes
$\{r^{\ell}_{a,j} : 0 \le a < (dt)^{d+2}\}$. Observe that, in
order to cover $B_{fill}$, it is necessary and sufficient to
select one node $r^{\ell}_{a,j}$ for each $\ell$ and $j$.
Furthermore, there is a way to do that such that each selected
$r^{\ell}_{a,j}$ covers one extra node in $B^{\ell}_{code}$
w.r.t. selected nodes in $R_{code}$.
Note that we somewhat abuse notation as we denote by $b^{\ell}_j$
vertices of $B_{fill}$, while we use $b^{\ell}_a$
for vertices of $B_{code}$, hence the only distinction is by the variable name.

%The following lemma summarizes the key properties of the enforcement graph.

\begin{lemma}
\label{lem:enforcement-exists}%
For any matrix $(\gamma_{\delta,\lambda})_{0 \le \delta < d+2,
0 \le \lambda < d}$ of size $(d+2)\times d$ with entries from
$\{0, \ldots, t-1\}$, there exists a set $\tilde{R}_{enf}
\subseteq R_{enf}$ of size
$k':=\frac{n}{2}(d^{d+2}-d)+(d+2)d(t-1)$, such that:
\begin{itemize}
\item[$\bullet$] each vertex in $B_{choice} \cup B_{fill}$
    has a neighbor in $\tilde{R}_{enf}$, and
\item[$\bullet$] for every $0 \le \ell < n/2$ we have
    $B^\ell_{code} \setminus
    N(\tilde{R}_{enf})=\{b^\ell_a: 0 \le \lambda < d, a =
    \sum_{0 \le \delta < d+2} (\lambda t +
    \gamma_{\delta,\lambda})(dt)^\delta\}$.
\end{itemize}
\end{lemma}

\begin{proof}
For each $0 \le \delta < d+2$ and $0 \le \lambda < d$, add to
$\tilde{R}_{enf}$ the set $\{r_{\delta,\lambda,\gamma} : 0 \le
\gamma < t, \gamma \neq \gamma_{\delta,\lambda}\}$ containing
$t-1$ vertices. Note that by construction $\tilde{R}_{enf}$
dominates the whole set $B_{choice}$. Consider a vertex
$b^\ell_a \in B^\ell_{code} \setminus N(\tilde{R}_{enf})$ and
observe that for each coordinate $0 \le \delta < d+2$, there
are exactly $d$ values that $a_\delta$ can have, where
$(a_0,\ldots,a_{d+1})$ is the $(dt)$-ary representation of $a$.
Indeed, for any $0 \le \delta < d+2$, we have $a_\delta \in
X_\delta = \{\lambda t + \gamma_{\delta,\lambda} : 0 \le
\lambda < d\}$, since otherwise $b^\ell_a$ would be covered by
$\tilde{R}_{enf}$ due to the $\delta$-th coordinate. Moreover
if we consider any $b^\ell_{(a_0,\ldots,a_{d+1})} \in
B^\ell_{code}$ such that $a_\delta \in X_\delta$ for $0 \le
\delta < d+2$, then $b^\ell_{(a_0,\ldots,a_{d+1})}$ is not
dominated by the vertices added to $\tilde{R}_{enf}$ so far.

Next, for each $\ell$ define $M^\ell:=\{b^\ell_a: 0 \le \lambda
< d, a = \sum_{0 \le \delta < d+2} (\lambda t +
\gamma_{\delta,\lambda})(dt)^\delta\}$ and observe that
$M^\ell$ are not dominated $\tilde{R}_{enf}$. For each $0 \le
\ell < n/2$, let $Z^\ell$ be the vertices of $B^\ell_{code}$
not yet covered by $\tilde{R}_{enf}$ and for each $1 \le j \le
d^{d+2}-d$ select exactly one distinct vertex $v_j \in Z^\ell
\setminus M^\ell$, where $v_j = b^\ell_a$, and add to
$\tilde{R}_{enf}$ the vertex $r^\ell_ {a,j}$. Observe that
after this operation $\tilde{R}_{enf}$ covers $B_{fill}$ and
moreover the only vertices of $B_{code}$ not covered by
$\tilde{R}_{enf}$ are the vertices of $\bigcup_{0 \le \ell <
n/2} M^\ell$. Since the total size of $\tilde{R}_{enf}$ equals
$d(d+2)(t-1)+\frac{n}{2}(d^{d+2}-d)$, the lemma follows. \qed
\end{proof}

\begin{lemma}
\label{lem:enforcement-forall}%
Consider an \RBDS instance $(H=(R\cup B,E),k)$ containing
$G_{enf}=(R_{enf}\cup B_{enf},E_{enf})$ as an induced subgraph,
with $R_{enf}\subseteq R$ and $B_{enf}\subseteq B$, such that
no vertex of $B_{choice} \cup B_{fill}$ has a neighbor outside
of $R_{enf}$. Then any feasible solution $\tilde{R}$ to $(H,k)$
contains at least $k':=\frac{n}{2}(d^{d+2}-d)+(d+2)d(t-1)$
nodes $\tilde{R}_{enf}$ of $R_{enf}$. Furthermore, for any
feasible solution $\tilde{R}$ to $(H,k)$ containing exactly
$k'$ vertices of $R_{enf}$, there exist a matrix
$(\gamma_{\delta,\lambda})_{0 \le \delta < d+2, 0 \le \lambda <
d}$ of size $(d+2)\times d$ with entries from $\{0, \ldots,
t-1\}$, such that for each $0 \le \ell < n/2$:
\begin{itemize}
\item[(a)] there are at least $d$ vertices in
    $U^\ell=B^\ell_{code} \setminus N(\tilde{R} \cap
    R_{enf})$, and
\item[(b)] $U^\ell$ is a subset of the $d^{d+2}$ nodes
    $b^{\ell}_{a}=b^{\ell}_{(a_0,\ldots,a_{d+1})}$ such
    that for each $\delta \in \{0,\ldots,d+1\}$ there
    exists $\lambda \in \{0,\ldots d-1\}$ with
    $a_{\delta}=\lambda t + \gamma_{\delta,\lambda}$.
\end{itemize}
\end{lemma}

\begin{proof}
Let $\tilde{R}$ be any feasible solution to $(H,k)$. Observe
that since $\tilde{R}$ dominates $B_{choice}$, for each $0 \leq
\delta < d+2$ and $0 \leq \lambda < d$ we have $|\tilde{R} \cap
\{r_{\delta,\lambda,\gamma} : 0 \leq \gamma < t\}| \geq t-1$.
Moreover in order to dominate vertices of $B_{fill}$, the set
$\tilde{R}$ has to contain at least $n/2(d^{d+2}-d)$ vertices
of $R_{fill}$. Consequently, if $\tilde{R}$ contains exactly
$k'$ vertices of $R_{enf}$, then for each $0 \leq \delta < d+2$ and
$0 \leq \lambda <d$, there is exactly one
$\gamma_{\delta,\lambda}$ such that
$r_{\delta,\lambda,\gamma_{\delta,\lambda}} \notin \tilde{R}$.
By the same argument as in the proof of
Lemma~\ref{lem:enforcement-exists}, we infer that for each
$\ell$, the set $B^\ell_{code} \setminus N(\tilde{R} \cap
R_{code})$ contains exactly $d^{d+2}$ vertices, and we denote
them as $U^\ell_0$. Observe that the set $\tilde{R} \setminus
R_{code}$ dominates at most $d^{d+2} -d$ vertices of
$U^\ell_0$, for each $0 \le \ell < n/2$, which proves
properties $(a)$ and $(b)$ of the lemma. \qed
\end{proof}

\subsection{The Overall Graph}

The construction of $H_{inst}=(R_{inst}\cup B_{inst},E_{inst})$
is rather simple. Let $(G_i=(V,E_i),\col_i)$ be the input \CPM
instances, with $0\leq i< T=t^{d(d+2)}$. By standard padding
arguments we may assume that all the graphs $G_i$ are defined
over the same set $V$ of even size $n$, i.e. $G_i=(V,E_i)$. For
each $v\in V$, we create a blue node $b_v\in B_{inst}$. For
each $e_i=\{u,v\}\in E_i$, we create a red node $r_{e,i}\in
R^i_{inst}$ and connect it to both $b_u$ and $b_v$. We let
$R_{inst}:=\bigcup_{0 \leq i <T}R^i_{inst}$. Intuitively, we
desire that a \RBDS solution, if any, selects exactly $n/2$
nodes from one set $R^i_{inst}$, corresponding to edges of
different colors, which together dominate all nodes $B_{inst}$:
This induces a feasible solution to \CPM for the $i$-th
instance.

It remains to describe the edges $E_{conn}$ which connect
$H_{enf}$ with $H_{inst}$. This is the most delicate part of
the entire construction. We map each index $i$, $0\leq i< T$,
into a distinct $(d+2)\times d$ matrix $M_i$ with entries
$M_i[\delta,\lambda]\in \{0,\ldots,t-1\}$, for all possible
values of $\delta$ and $\lambda$. Consider an instance $G_i$.
We connect $r_{e,i}$ to $b^{\ell}_{a}$ iff $\ell=\col_i(e_i)$
and there exists $0 \le \lambda < d$ such that the expansion
$(a_0,\ldots,a_{d+1})$ of $a$ in base $dt$ satisfies
$a_\delta=M_i[\delta,\lambda]+\lambda\cdot t$ on each
coordinate $0 \le \delta < d+2$. The final graph $H:=(R \cup
B,E)$ we construct for our instance \RBDS is then given by
$R:=R_{inst} \cup R_{enf}$ and $E:=E_{inst} \cup E_{enf} \cup
E_{conn}$. See Fig.~\ref{fig:ds}.

\begin{figure}
\begin{center}
\includegraphics[scale=1]{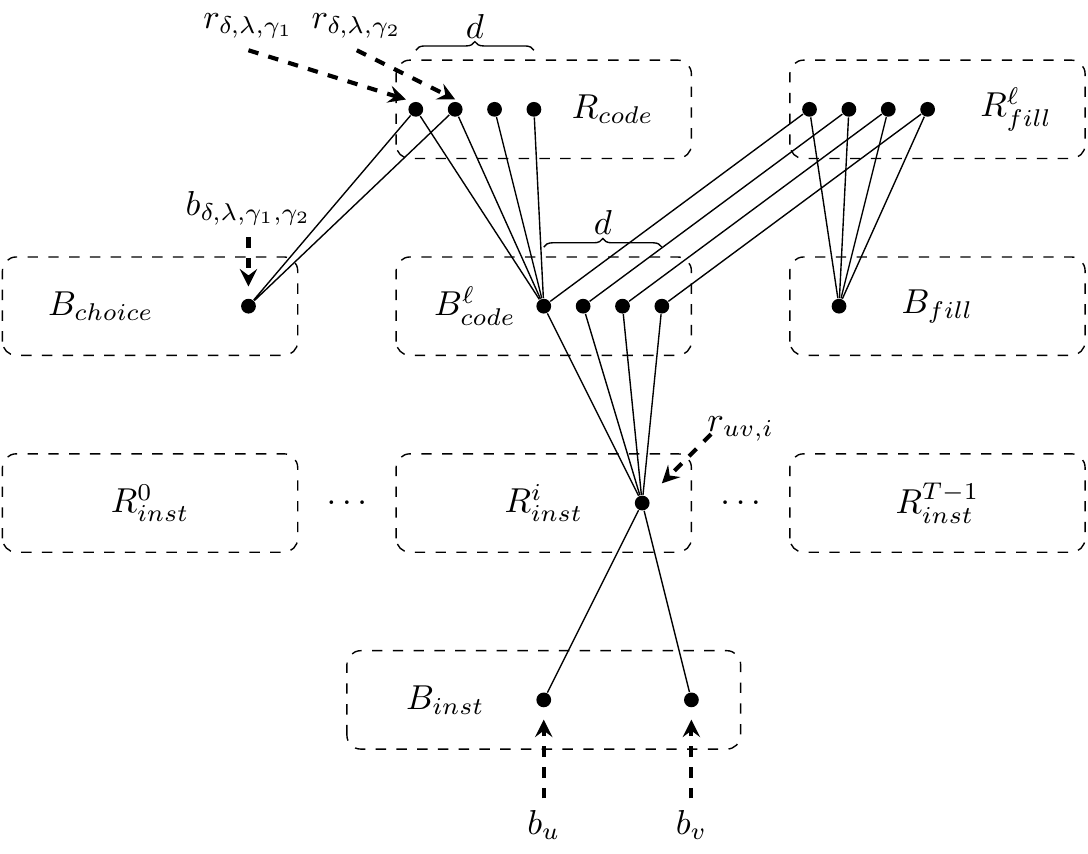}
\caption{Construction of the graph $H$. For simplicity the figure does not include sets $R^{\ell'}_{fill}$ and $B^{\ell'}_{code}$ for $\ell'\neq \ell$.}
\label{fig:ds}
\end{center}
\end{figure}

%We start by bounding the degeneracy of the overall graph $G$.
\begin{lemma}
\label{lem:ds-deg}%
$H$ is $(d+2)$-degenerate.
\end{lemma}

\begin{proof}
Observe that each vertex of $\bigcup_{0 \le i < T} R^i_{inst}$
is of degree exactly $d+2$ in $H$, since it is adjacent to
exactly two vertices of $B_{inst}$ and exactly $d$ vertices of
the enforcement gadget,
so we put all those vertices first to our ordering. Next, we
take vertices of $B_{inst}$, as those have all neighbors
already put into the ordering. Therefore it is enough to argue
about the $(d+2)$-degeneracy of the enforcement gadget. We
order vertices of $R_{fill} \cup B_{choice}$, since those are
of degree exactly two in $H$. In $H \setminus R_{fill}$ the
vertices of $B_{fill}$ become isolated, so we put them next to
our ordering. We are left with the vertices of the encoding
gadget. Observe, that each blue vertex of the encoding gadget
has exactly $d+2$ neighbors in $R_{code}$, one due to each
coordinate, hence we put the vertices of $B_{code}$ next and
finish the ordering with vertices of $R_{code}$. \qed
\end{proof}

\begin{lemma}
\label{lem:ds-mapping}%
Let $k:=(d+2)d(t-1)+n/2(d^{d+2}-d)+n/2=k'+n/2$. Then $(H,k)$ is
a YES-instance of \RBDS iff $(G_i,\col_i)$ is a YES-instance of
\CPM for some $i \in \{0,\ldots,T-1\}$.
\end{lemma}

\begin{proof}
Let us assume that for some $i_0$ the instance
$(G_{i_0},\col_{i_0})$ is a YES-instance and $E' \subseteq
E_{i_0}$ is the corresponding solution. We use
Lemma~\ref{lem:enforcement-exists} with the matrix $M_{i_0}$
assigned to the instance $i_0$ to obtain the set
$\tilde{R}_{enf}$ of size $(d+2)d(t-1)+\frac{n}{2}(d^{d+2}-d)$.
As the set $\tilde{R}$ we take $\tilde{R}_{enf} \cup
\{r_{e,i_0} : e\in E'\}$. Clearly $|\tilde{R}| = k$. Since $E'$
is a perfect matching, $\tilde{R}$ dominates $B_{inst}$. By
Lemma~\ref{lem:enforcement-exists}, $\tilde{R}$ dominates
$B_{fill} \cup B_{choice}$ and all but $d$ vertices of each
$B^\ell_{code}$, so denote those $d$ vertices by $M^\ell$.
Consider each $0 \le \ell < n/2$, and observe that since $E'$
is multicolored and by the construction of $H$, the set of
neighbors of $r_{e,i_0}$ in $B_{code}$ is exactly
$M^{\col_{i_0}(e)}$; and hence $\tilde{R}$ is a solution for
$(H,k)$.

In the other direction, assume that $(H,k)$ is a YES-instance
and let $\tilde{R}$ be a solution of size at most $k$. By
Lemma~\ref{lem:enforcement-forall}, the set $\tilde{R}$
contains at least $k'=\frac{n}{2}(d^{d+2}-d)+(d+2)d(t-1)$
vertices of $R_{enf}$ and since $\tilde{R}$ needs to dominate
also $B_{inst}$ it contains at least $\frac{n}{2}$ vertices of
$\bigcup_{0 \le i < T} R^i_{inst}$, since no vertex of $H$
dominates more than two vertices of $B_{inst}$. Consequently
$|\bigcup_{0 \le i < T} R^i_{inst} \cap \tilde{R}| = n/2$ and
$|R_{enf} \cap \tilde{R}|=k'$. We use
Lemma~\ref{lem:enforcement-forall} to obtain a matrix
$M=(\gamma_{\delta,\lambda})$ of size $(d+2)\times d$.
Moreover, by property (a) of
Lemma~\ref{lem:enforcement-forall}, there are at least $d$
vertices in $U^\ell$, and consequently for each color $\ell$
the set $\tilde{R}$ contains exactly one vertex of the set
$\{r_{e,i}: 0 \le i < T, \col_i(e) = \ell\}$. Our goal is to
show that for each $0 \le i < T$, such that a matrix different
than $M$ is assigned to the $i$-th instance, we have $\tilde{R}
\cap R^i_{inf}=\emptyset$, which is enough to finish the proof
of the lemma. Consider any such $i$ and assume that the
matrices $M_i$ and $M$ differ in the entry
$M_i[\delta',\lambda'] \neq \gamma_{\delta',\lambda'}$. Let
$\ell$ be a color such that $r_{e,i} \in \tilde{R}$ and
$\col_i(e) = \ell$. By property (b) of
Lemma~\ref{lem:enforcement-forall}, the set of at least $d$
vertices of $B^\ell_{code}$ not dominated by $\tilde{R} \cap
R_{enf}$ is contained in
$U^\ell_0=\{b^\ell_{(a_0,\ldots,a_{d+1})} : \forall_{0 \le
\delta < d+2} \text{ if } \lambda t \le a_\delta < (\lambda+1)t
\text{ then } a_\delta=\lambda t+\gamma_{\delta,\lambda}\}$.
However, by our construction of edges of $H$ between
$R^i_{inst}$ and $B^\ell_{code}$, we have $(N_H(r_{e,i}) \cap
B^\ell_{code}) \not\subseteq U^\ell_0$ since the vertex
$b^\ell_{(a_0,\ldots,a_{\delta+1})} \in N_H(r_{e,i}) \cap
B^\ell_{code}$ with $a_\delta=\lambda' t +
M_i[\delta',\lambda']$ does not belong to $U^\ell_0$ and
consequently does not belong to $U^\ell$, which leaves at least
one vertex of $B^\ell_{code}$ not dominated by $\tilde{R}$; a
contradiction. \qed
\end{proof}

Lemmas \ref{lem:ds-deg} and \ref{lem:ds-mapping} imply that,
for any $d \ge 1$, there exists a weak $d(d+2)$-composition
from \CPM to \RBDS in $(d+2)$-degenerate graphs. The proof of
Theorem~\ref{thm:ds} thus follows from Lemmas~\ref{Lemma:
d-composition gives lower bounds} and~\ref{lem:rbds}.

%%%%%%%%%%%%%%%%%%%%%%%%%%%%%%%%%%%%%%%%%%%%%%%%%%%%%%%%%%%%%%%%%%%%%
%%%%%%%%% Section: Independent Dominating Set
%%%%%%%%%%%%%%%%%%%%%%%%%%%%%%%%%%%%%%%%%%%%%%%%%%%%%%%%%%%%%%%%%%%%%

\section{Independent Dominating Set}
\label{Section: Independent Dominating Set}

The \IDSlong (\IDS) problem is the variant of \DS where we
require the dominating set $S$ to induce an independent set
(i.e. nodes in $S$ have to be pairwise non-adjacent). We next
describe a weak $d$-composition from \tESC, to \IDS in
$(d+4)$-degenerate graphs. The input of \tESC is a set system
$(U,\cF)$, where each set in $\cF$ contains exactly three
elements and the question is whether there is a collection $\cS
\subseteq \cF$ of disjoint sets which partition $U$, i.e.
$\bigcup \cS = U$. The \tESC problem is NP-complete by a
reduction from \TDM, which is NP-complete due to
Karp~\cite{karp-21}.
%,
%hence our composition together with Lemma~\ref{Lemma: d-composition gives lower bounds} gives the desired lower
%bound of Theorem~\ref{Theorem: Main}.

Consider a fixed value of $d \ge 1$ and let $(\cF_0,U_0),
\ldots, (\cF_{T-1},U_{T-1})$ be $T := t^d$ instances of \tESC.
Without loss of generality we can assume that in each
instance the same universe $U$ of size $n$ is used, such that
$n \equiv 0 \pmod 3$, since if for some $i$ we have $|U_i| \neq
0 \pmod 3$, then $(\cF_i, U_i)$ is clearly a NO-instance, and
moreover we can pad each universe with additional triples to
make sure that each $U_i$ has the same cardinality. Therefore
we assume that for each $i=0,\ldots,T-1$ we have $U_i = U$ and
$|U| = n$.

\begin{figure}
\begin{center}
\includegraphics{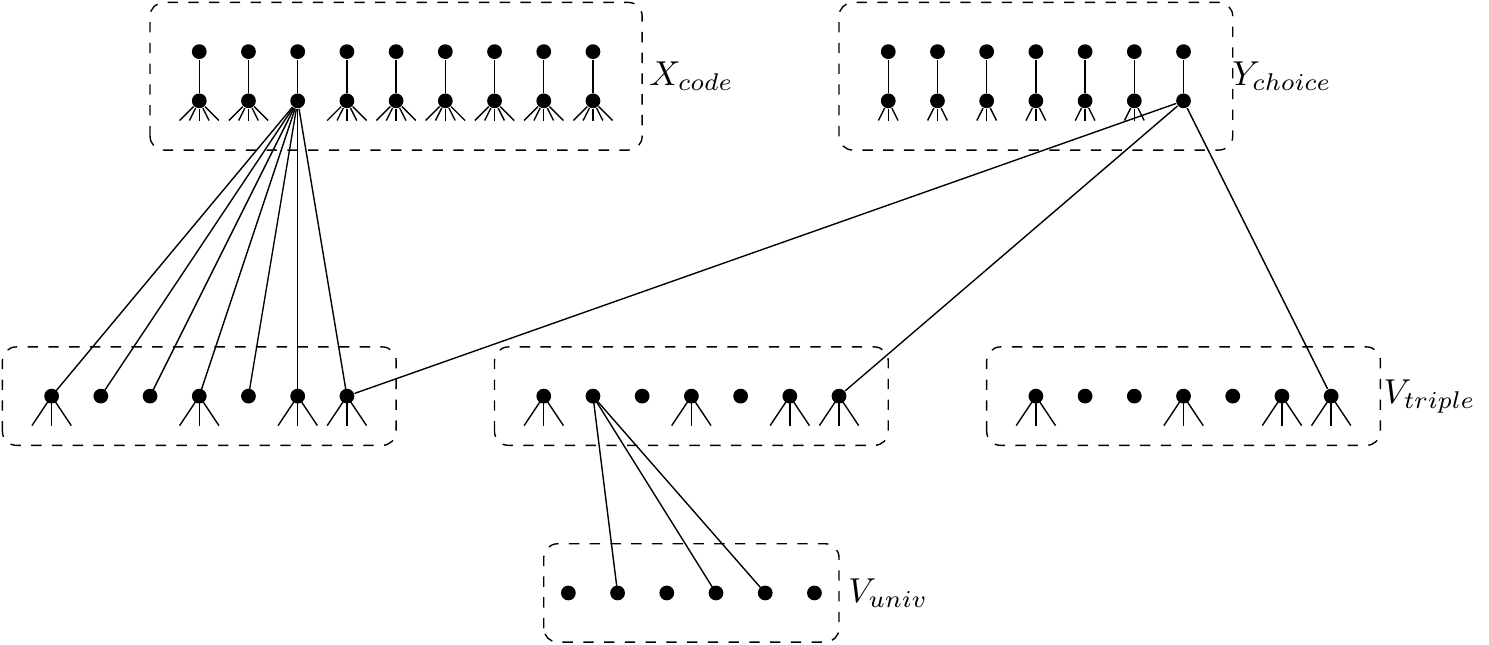}
\caption{The construction of the graph $H$ in the composition for \IDSlong.}
\label{fig:ids}
\end{center}
\end{figure}

We construct an instance $(H=(V,E),k)$ of \IDS for a properly
chosen parameter $k$. Similar to construction in
Section~\ref{Section: Dominating Set}, the graph $H$ will
contain two graphs, the {\em enforcement graph}
$H_{enf}=(V_{enf},E_{enf})$ and the {\em instance graph}
$H_{inst}=(V_{inst},E_{inst})$, plus some edges $E_{conn}$
connecting them (see Fig.~\ref{fig:ids}). The graph $H_{inst}$ contains the node set $V_{univ}:=\{v_u:
u\in U\}$ (i.e., one node per element of the universe).
Furthermore, it contains a node set $V_{triple}$ which includes
one node $v_{i,S}$ for every instance index $0\leq i<T$ and
every triple of elements of the universe $S\in \binom{U}{3}$.
We connect nodes $v_u$ and $v_{i,S}$ iff $S \in \cF_i$ and
$u\in S$. Observe that there might be some isolated nodes in
$V_{triple}$. The graph $H_{enf}$ consists of two induced
matchings. It contains an \emph{encoding matching} $X_{code}$
with edges $\{x_{\gamma,\delta},x'_{\gamma,\delta}\}$ for all
integers $0 \le \gamma < t$ and $0 \le \delta < d$.
Furthermore, it contains a \emph{choice matching} $Y_{choice}$
with edges $\{y_{S},y'_{S}\}$ for all triples $S\in
\binom{U}{3}$.
It remains to describe $E_{conn}$. First of all, we connect
each $y_{S}$ to every $v_{i,S}$. Second, we connect
$x_{\gamma,\delta}$ with $v_{i,S}$ iff $i_{\delta}=\gamma$
where $(i_0,\ldots,i_{d-1})$ is the $t$-ary expansion of index
$i$, \emph{i.e.} $i=\sum_{\delta=0}^{d-1}i_\delta \cdot
t^\delta$ with $0 \leq i_\delta < t$.

\begin{lemma}
\label{lem:ids-deg}
The graph $H$ is $d+4$ degenerate.
\end{lemma}
\begin{proof}
Consider any ordering of vertices where we put vertices of
$V_{triple}$ first, and the remaining vertices. Observe, that
each vertex of $V_{triple}$ is of degree exactly $d+4$ in $H$,
since it has exactly $d$ neighbors in $X_{code}$, exactly one
neighbor in $Y_{choice}$ and exactly three neighbors in
$V_{univ}$. After removing $V_{triple}$, all vertices have
degree at most $1$. The claim follows. \qed
\end{proof}

\begin{lemma}
\label{lem:ids}%
Let $k:=dt + \binom{n}{3} + n/3$. Then $(H,k)$ is a
YES-instance of \textsc{IDS} iff there exists $0 \le j < T$,
such that $(\cF_{j},U)$ is a YES-instance.
\end{lemma}
\begin{proof}
For the if part, let us assume that for some $0 \le j < T$ the
instance $(\cF_{j},U)$ of \tESC is a YES-instance and let $\cS
\subseteq \cF_{j}$ be a solution, i.e. a collection of $n/3$
disjoint sets. Construct an independent dominating set $D$ of
exactly $k$ vertices as follows. Let $(j_0,\ldots,j_{d-1})$ be
the $t$-ary expansion of $j$. For $0 \le \gamma < t, 0 \le
\delta < d$ if $j_\delta = \gamma$, then add
$x'_{\gamma,\delta}$ to $D$ and otherwise add to $D$ the vertex
$x_{\gamma,\delta}$. For $S \in \binom{U}{3}$, add $y'_S$ to
$D$ if $S \in \cS$ and add $y_S$ to $D$ otherwise. Finally, add
to $D$ the $n/3$ vertices $v_{j,S}$ with $S \in \cS$. Clearly
$|D|=k$. Moreover $D$ is an independent set. In fact,
$H[X_{code} \cup Y_{choice}]$ is a matching and we have taken
exactly one endpoint of each one of its edges. Moreover for
each $v_{j,S} \in D$ by construction there is no edge between
$v_{j,S}$ and the remaining vertices of $D$. To prove that $D$
is a dominating set observe that all the vertices of $V_{univ}$
are dominated because $\cS$ is a solution for $(\cF_{j},U)$,
and all the vertices of $X_{code} \cup Y_{choice}$ are
dominated because $D$ contains exactly one endpoint of each
edge of the matching $H[X_{code} \cup Y_{choice}]$. Finally,
each vertex $v_{i,S}$ is dominated, since either $i\neq j$ and
then $v_{i,S}$ is dominated by $X_{choice} \cap D$ due to the
coordinate at which $i$ and $j$ differ, or $S \in \cS$ and then
$v_{j,S} \in D$ or $y_S \in D$.

For the only if part, let $D$ be a dominating set in $H$ of
size at most $k$. Observe that since each vertex
$x'_{\gamma,\delta}$ and $y'_S$ is of degree exactly one in
$H$, we have $|D \cap X_{code}| \ge dt$ and $|D \cap
Y_{choice}| \ge \binom{n}{3}$. Moreover, since no vertex has
more than three neighbors in $V_{univ}$ we have $|D \cap
(V_{univ} \cup V_{triple})| \ge n/3$. Therefore, all the three
mentioned inequalities are tight. Moreover $D$ contains exactly
$n/3$ vertices of $V_{triple}$, since if $D$ would contain a
vertex of $V_{univ}$ than $|D \cap (V_{univ} \cup V_{triple})|$
would be strictly greater than $n/3$. Define $\cS = \{S \in
\binom{U}{3}\, |\, \exists 0 \le i < T: v_{i,S} \in D\}$.
Observe, that since $|D \cap V_{triple}| = n/3$ we have $|\cS|
= n/3$. We want to show that there exists $0 \le j < T$, such
that $\cS \subseteq \cF_{j}$, which is enough to prove that
$(\cF_{j},U)$ is a YES-instance. Assume the contrary. Then
there exist two indices $0 \le i_1 < i_2 < T$, such that there
exist $S_1,S_2 \in \cS$ with $v_{i_1,S_1} \in D$ and
$v_{i_2,S_2} \in D$. Since $D$ is an independent set, it means
that no vertex of $V':=\{v_{i,S} : i\in \{i_1,i_2\}, S \in
\cS\}$ has a neighbor in $D \cap X_{code}$ nor in $D \cap
Y_{choice}$ and therefore $V' \subseteq D$. However $|V'| =
2|\cS| = 2n/3 > n/3$, a contradiction. \qed
\end{proof}

\begin{theorem}
Let $d \geq 4$. Then \textsc{IDS} has no kernel of size
$O(k^{d-4-\varepsilon})$ for any constant $\varepsilon>0$
unless \textnormal{coNP} $\subseteq$ \textnormal{NP/poly}.
\end{theorem}
\begin{proof}
From Lemmas \ref{lem:ids-deg} and \ref{lem:ids}, there exists a
weak $d$-composition from \tESC to \IDS in $(d+4)$-degenerate
graphs for any $d\geq 1$. The claim thus follows from
Lemma~\ref{Lemma: d-composition gives lower bounds}.
\end{proof}

%%%%%%%%%%%%%%%%%%%%%%%%%%%%%%%%%%%%%%%%%%%%%%%%%%%%%%%%%%%%%%%%%%%%%
%%%%%%%%% Section: Induced Matching
%%%%%%%%%%%%%%%%%%%%%%%%%%%%%%%%%%%%%%%%%%%%%%%%%%%%%%%%%%%%%%%%%%%%%

\section{Induced Matching}
\label{Section: Induced Matching}

In this section we show a kernelization lower bound for the
{\sc Induced Matching} (\IM) problem in $d$-degenerate graphs.
In \IM, the input is a graph $G$ and an integer $k$, and the
goal is to determine whether there exists a set of $k$ edges
$e_1,\ldots,e_k$ in $G$ such that there is no edge in $G$
connecting two endpoints of $e_i$ and $e_j$ for all $i \neq j
\in \{1,\ldots,k\}$. The main result of this section is given
by the following theorem.

\begin{theorem}
\label{Theorem: IM}%
Let $d \geq 3$. Then \textsc{IM} has no kernel of size
$O(k^{d-3-\varepsilon})$ for any constant $\varepsilon>0$
unless \textnormal{coNP} $\subseteq$ \textnormal{NP/poly}.
\end{theorem}

For our kernel lower bound on \IM, we present a weak
$d$-composition from the \MC\ problem in which the input is a
graph $G:=(V,E)$ and a vertex-coloring $\col: V \to
\{1,\ldots,k\}$, and the goal is to determine whether there
exists a multicolored clique of size $k$ in $G$, that is,
whether there exists a set of pairwise adjacent vertices
$v_1,\ldots,v_k$ in $G$ with $\col(v_i) \neq \col(v_j)$ for all
$i \neq j \in \{1,\ldots,k\}$. It is well known that \MC is
\NP-hard~\cite{Fellows-et-al2009}.

Let $(G_i=(V_i,E_i),\col_i)$, $0\leq i <T:=t^d$, be the input
instances of \MC. By standard padding and vertex-renaming
arguments, we can assume that all graphs $G_i$ are defined over
the same vertex set $V$ of size $n$, and that each vertex $v
\in V$ is assigned the same color $\col(v) \in \{1,\ldots,k\}$
by all coloring functions $\col_i$'s (note that this can be
done even if the number of colors in each graph is different).
We can further assume that for each $\{u,v\} \in E_i$ we have
$\col(u) \neq \col(v)$, since all edges between vertices of the
same color can never appear in any multicolored clique.
Finally, we also assume that $\binom{k}{2} -k > d$, since
otherwise a weak $d$-composition can trivially be constructed
by solving each instance separately in polynomial time.

\begin{figure}[t]
\begin{center}
\includegraphics{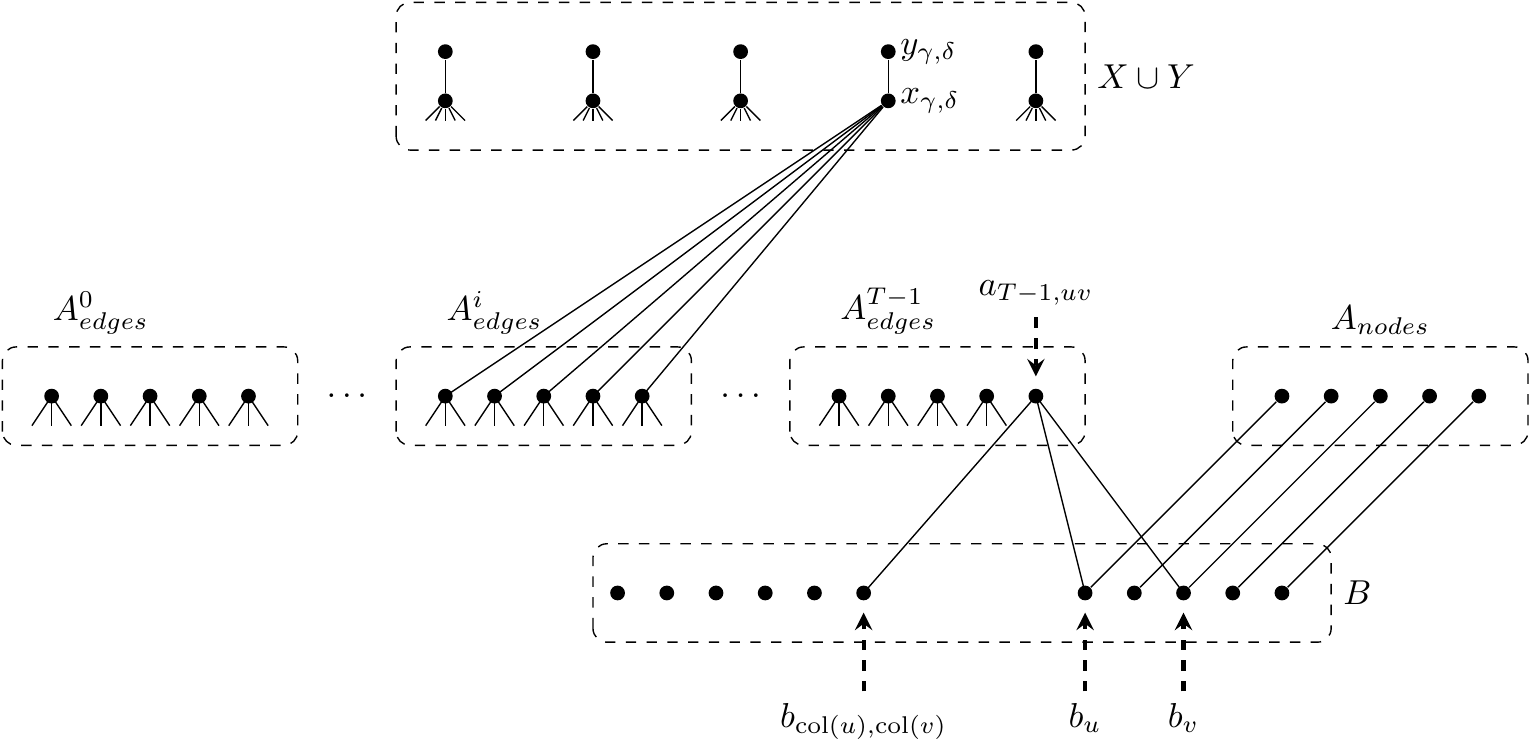}
\caption{The construction of the graph $H$ in the composition for \IMlong.}
\label{fig:im}
\end{center}
\end{figure}

We next construct and instance $(H=(W,F),k')$ of \IM for a
proper parameter $k'$. As in previous sections, $H$ consists of
an \emph{instance graph} $H_{inst}=(W_{inst},F_{inst})$ and an
\emph{enforcement gadget} $(H_{enf},E_{conn})$, where $H_{enf}$
is a graph and $E_{conn}$ is a proper set of edges between
$H_{inst}$ and $H_{enf}$ (see Fig.~\ref{fig:im}). The role of the instance graph is the
guarantee that feasible solutions to any \MC instance induce
feasible solutions to the \IM instance, and the enforcement
gadget ensures that we cannot combine partial solutions of \MC
instances to obtain a feasible solution to the \IM instance.

Graph $H_{inst}=(A\cup B,F_{inst})$ is bipartite, with
$A=A_{nodes}\cup A_{edges}$ and $B=B_{nodes}\cup
B_{col-pairs}$. Sets $A_{nodes}$ and $B_{nodes}$ contain a node
$a_v$ and $b_v$, respectively, for each $v\in V$. We have $A_{edges} = \bigcup_{0 \le i < T} A^i_{edges}$,
where the set $A^i_{edges}$ contains a node $a_{i,e}$ for each instance $i$ and
$e_i\in E_i$. Set $B_{col-pairs}$ contains a node
$b_{\alpha,\beta}\in B$ for every pair of colors $1\leq
\alpha<\beta\leq k$. Set $F_{inst}$ contains all the edges of
type $\{a_v,b_v\}$, plus edges between each
$a_{i,e}=a_{i,uv}$ and nodes $b_u$, $b_v$ and
$b_{\col(u),\col(v)}$. 

\begin{lemma}
\label{Lemma: IM1}%
Let $M'$ be an induced matching in $H_{inst}$. Then $|M'|\leq
\binom{k}{2} + n- k$. Moreover, equality holds iff the graph
with vertex set $V'=\big\{v \in V \,:\, b_v \in V(M')\big\}$
and edge set $E':=\big\{e \,:\, a_{i,e} \in V(M') \text{ for
some } i \in \{0,\ldots,T-1\}\big\}$ is a multicolored clique
of the graph $(V,\cup_i E_i)$.
\end{lemma}
\begin{proof}
Let $M'$ be a maximum induced matching in $H_{inst}$. Clearly
$n \leq |M'| \leq n+\binom{k}{2}$ (the lower bound comes from
the matching induced by pairs $(a_v,b_v)$ ). If $|M'| = n$ then
we are done, so assume $|M'| > n$. Since $\binom{k}{2} - k > d
>0$, we have $\binom{k}{2} + n- k > n$, and so some vertex
$b_{\alpha,\beta}$ must be matched in $M'$. Consider some edge
$m \in M'$ which includes some vertex $b_{\alpha,\beta}$. Then
by construction, $m=\{b_{\alpha,\beta},a_{i,e}\}$ where
$a_{i,e}$ corresponds to an edge $e=\{u,v\} \in E_i$, for some
$i \in \{0,\ldots,T-1\}$, where $\col(u)=\alpha$ and
$\col(v)=\beta$. As $\{b_u,a_{i,e}\}, \{b_{v},a_{i,e}\} \in
F_{inst}$, the vertices $b_u$ and $b_v$ cannot be matched in
$M'$. Thus, for each vertex $b_{\alpha,\beta}$ that is matched
in $M'$, a vertex of $\{b_u: u \in V, \col(u) = \alpha\}$ is
not matched in $M'$, and a vertex of $\{b_v: v \in V, \col(v) =
\beta\}$ is not matched in $M'$. Since a vertex can only appear
in at most $k$ color pairs, we have $M' \leq \binom{k}{2} + n-
k$. Furthermore, if a vertex $v \in V$ is not contained in any
edge of $E'$, then $M' \cup \{a_v,b_v\}$ is also an induced
matching. It is now not difficult to see that $|M'| =
\binom{k}{2} + n- k$ iff $|V \setminus V'| = n-k$, and
$(V',E')$ is a multicolored clique of size $k$. \qed
\end{proof}

As Lemma~\ref{Lemma: IM1} suggests, to ensure that maximum
induced matchings in $H_{inst}$ are meaningful to us, we need
to guarantee that only nodes $c_{i,e}$ for a unique index $i$
are matched. To this aim, we introduce the following
enforcement gadget $(H_{enf},E_{conn})$. Graph $H_{enf}=(X\cup
Y,E_{enf})$ is bipartite. Sets $X$ and $Y$ contain a node
$x_{\gamma,\delta}$ and $y_{\gamma,\delta}$, respectively, for
all integers $0\leq \gamma <t$ and $0\leq \delta <d$. We add
edges between all pairs
$\{x_{\gamma,\delta},y_{\gamma,\delta}\}$. Finally, we connect
$x_{\gamma,\delta}$ with $a_{i,e}$ iff $i_{\delta}=\gamma$
where $(i_0,i_1,\ldots,i_{d-1})$ is the $t$-ary expansion of
index $i$. Observe that each $a_{i,e}$ is adjacent to exactly
$d$ distinct vertices  $x_{\gamma,\delta}$.

\begin{lemma}
\label{Lemma: IM2}%
$H$ is $(d+3)$-degenerate.
\end{lemma}
\begin{proof}
Consider any ordering of the vertices which places first nodes
$A_{nodes}$ and $Y$, then node $A_{edges}$, then nodes $B$, and
finally nodes $X$. It is not difficult to check that each
vertex is adjacent to at most $d+3$ vertices appearing to its
right in this ordering. \qed
\end{proof}

\begin{lemma}
\label{Lemma: IM3}%
Let $k':=t(d-1)+\binom{k}{2}+n-k$. Then $(H,k')$ is a
YES-instance of \textsc{IM} iff $(G_i,\col_i)$ is a
YES-instance of \MC for some index $i \in \{0,\ldots,T-1\}$.
\end{lemma}
\begin{proof}
Suppose some $G_i$ has a multicolored clique of size $k$. Then
by Lemma~\ref{Lemma: IM1} we can find an induced matching $M'$
of size $\binom{k}{2}+n-k$ in $H_{inst}$ which matches only
vertices of type $a_{i,e}$. Let us add all the edges
$\{x_{\gamma,\delta},y_{\gamma,\delta}\}$ such that $i_\delta
\neq \gamma$. There are precisely $t(d-1)$ such edges, which
together with the edges of $M'$, form an induced matching of
size $k'$ in $H$.

For the converse direction, suppose $M$ is an induced matching
of size $k'$ in $H$. First observe that we can assume $M$ does
not contain any edges of $F(A_{edge},X)$ since any such edge
$\{a_{i,e},x_{\gamma,\delta}\}$ can be safely replaced with the
edge $\{x_{\gamma,\delta},y_{\gamma,\delta}\}$. Now as
$\binom{k}{2} - k > d$ w.l.o.g., the matching $M$ must include
some edge of $F(B,A_{edge})$ since there are $n+td < k'$ edges
altogether in $F(A_{node},B)\cup F(X,Y)$. So let $a_{i,e}$ be a
vertex of $A_{edge}$ which is matched in $M$. Then as $a_{i,e}$
is matched in $M$, this means that $d$ vertices of $X$ cannot
be matched in $M$, precisely those vertices $x_{\gamma,\delta}$
with $i_\delta = \gamma$. Thus, $|F(X,Y) \cap M| \leq t(d-1)$.
Since $M \setminus F(X,Y)$ is a subset of edges in $H_{inst}$,
and since the maximum induced matching in $H_{inst}$ is at most
$\binom{k}{2} + n - k$ by Lemma~\ref{Lemma: IM1}, this implies
that $M$ contains exactly $\binom{k}{2} + n - k$ edges of
$H_{inst}$ and $t(d-1)$ edges of $F(X,Y)$. By construction of
$F(A_{edge},X)$, the latter assertion implies that there exists
some $i$ such that each vertex of $A_{edge}$ that is matched by
$M$ is of the form $a_{i,e}$ for some $e \in E_i$. By
Lemma~\ref{Lemma: IM1}, the second assertion implies that $\{v
\in V : \{b_v,a_{i,e}\} \in M \text{ for some } e \in E_i\}$ is
a multicolored clique of size $k$ in~$G_i$. \qed
\end{proof}

Theorem~\ref{Theorem: IM} now directly follows from
Lemmas~\ref{Lemma: IM2}, \ref{Lemma: IM3}, and~\ref{Lemma:
d-composition gives lower bounds}.

%%%%%%%%%%%%%%%%%%%%%%%%%%%%%%%%%%%%%%%%%%%%%%%%%%%%%%%%%%%%%%%%%%%%%
%%%%%%%%% Section: Upper Bounds
%%%%%%%%%%%%%%%%%%%%%%%%%%%%%%%%%%%%%%%%%%%%%%%%%%%%%%%%%%%%%%%%%%%%%

\section{Upper Bounds}

%In the following section we consider the \ConVC and \CapVC
%problems. Both problems are known not have polynomial kernels
%in general graphs~\cite{DomLokshtanovSaurabh2009}, and for
%\ConVC an upper bound of $O(k^{d+1})$ and lower bound of
%$O(k^{d+1-\varepsilon})$ was given for $d$-degenerate graphs
%in~\cite{Cygan-et-al2012}. We improve the upper bound
%of~\cite{Cygan-et-al2012} for \ConVC\ to $O(k^d)$, and show an
%upper bound of$O(k^d)$ and lower bound of
%$O(k^{d-1-\varepsilon})$ for \CapVC. 
Both upper bounds that we present rely on the following easy lemma.

\begin{lemma}
\label{Lemma: Easy}%
Let $G:=(A \cup B,E)$ be a bipartite $d$-degenerate graph where
all vertices in $B$ have degree greater than $d$. Then $|B|
\leq d|A|$.
\end{lemma}

\begin{proof}
By the $d$-degeneracy of $G$, we know that $|E| \leq
d(|A|+|B|)$. Since each vertex of $B$ has degree greater than
$d$, we also know that $(d+1)|B| \leq |E|$. Subtracting $d|B|$
from both inequalities gives $|B| \leq d|A|$. \qed
\end{proof}

\subsection{Dominating and independent dominating set}

\subsection{Connected and capacitated vertex cover}

Let us begin with the kernel for \ConVC. In \ConVC, our goal is
to determine whether a given graph $G$ has a set of vertices
$A$ such that $G[V(G)\setminus A]$ is edgeless ($A$ is a
\emph{vertex cover}) and $G[A]$ is connected ($A$ is
\emph{connected}). Our kernelization algorithm uses two simple
reduction rules which are given below, the second of which is a
variant of the well-known \emph{crown reduction
rule}~\cite{Fellows2003}. We say that a set of vertices $S
\subseteq V(G)$ is a set of \emph{twins} in $G$ if $N(v)=N(u)$
for all $u,v \in S$ (note that this implies that $S$ is an
independent set in $G$). We let $N(S)$ denote the common set of
neighbors of $S$.

\begin{reduction-rule}
\label{Rule: Isolated}%
If $G$ has an isolated vertex remove it.
\end{reduction-rule}

\begin{reduction-rule}
\label{Rule: ConVC}%
If $S \subseteq V(G)$ is a subset of at most $d+1$ twin
vertices with $|N(S)|< |S|$, remove an arbitrary vertex of $S$
from~$G$.
\end{reduction-rule}

\begin{lemma}
\label{Lemma: ConVC}%
Let $G$ be a graph, and let $G'$ be a graph resulting from
applying either Rule~\ref{Rule: Isolated} or Rule~\ref{Rule:
ConVC} to $G$. Then for any integer $k$, the graph $G$ has a
connected vertex cover of size $k$ iff $G'$ has a connected
vertex cover of size $k$.
\end{lemma}

\begin{proof}
The lemma is obvious for Rule~\ref{Rule: Isolated}, so let us
focus on Rule~\ref{Rule: ConVC}. Observe that there exists a
minimal size connected vertex cover $A$ in $G$ with $N(S)
\subseteq A$. Furthermore, $A$ contains at most one vertex of
$S$, and this vertex can be replaced by any other vertex of $S$
to obtain another connected vertex cover for $G$ of equal size.
Replacing this vertex (if it exists) with a vertex of $S \cap
V(G')$, we obtain a connected vertex cover $A'$ for $G'$ with
$|A'|=|A|$. Conversely, using similar arguments one can
transform a minimum size connected vertex cover for $G'$ to an
equal size connected vertex cover for $G$. \qed
\end{proof}

\begin{theorem}
\label{Theorem: ConVC}%
\ConVC\ in $d$-degenerate graphs has a kernel of size $O(k^d)$.
\end{theorem}

\begin{proof}
Our kernelization algorithm for \textsc{Connected Vertex Cover}
in $d$-degenerate graphs exhaustively applies Rule~\ref{Rule:
Isolated} and Rule~\ref{Rule: ConVC} until they no longer can
be applied. Since Rule~\ref{Rule: Isolated} can be implemented
in linear time, and Rule~\ref{Rule: ConVC} can be done in
$n^{O(d)}$ time, this algorithm runs in polynomial time. Let
$G'$ be the graph resulting from the kernelization. Observe
that both reduction rules that were used do not increase the
degeneracy of the graph, and so $G'$ is $d$-degenerate as well.
Furthermore, due to Lemma~\ref{Lemma: ConVC}, we know that $G$
has a connected vertex cover of size $k$ iff $G'$ has one as
well. We next show that $|V(G')| =O(k^d)$, or otherwise $G'$
has no connected vertex cover of size $k$.

Suppose that $G'$ has a connected vertex cover $A$ of size $k$.
Then as $A$ is a vertex cover, the set $B:=V(G) \setminus A$ is
an independent set in $G$. For $i:=0,\ldots,d$, define $B_i
\subseteq B$ to be set of all vertices in $B$ with degree $i$
in $G$, and let $B_{>d} \subseteq B$ be the vertices in $B$
with degree greater than $d$ in $G$. Then $B:= B_0 \cup \cdots
\cup B_d \cup B_{>d}$, and $|B_0|=0$ since Rule~\ref{Rule:
Isolated} cannot be applied. Due to Rule~\ref{Rule: ConVC}, for
each subset of $i$ vertices $A' \subseteq A$, $1 \leq i \leq
d$, there are at most $i$ vertices $B' \subseteq B_i$ with
$N(B')=A'$. We therefore have $|B_i| \leq i\binom{k}{i}$ for
each $i \in \{1,\ldots,d\}$, and $\sum_{i:=0,\ldots,d} |B_i|
\leq d k^d$. Furthermore, we also have $|B_{>d}| \leq dk$ by
applying Lemma~\ref{Lemma: Easy} to the bipartite graph on $A$
and $B_{>d}$. Accounting also for $A$, we get
$$
|V(G')|=|A|+|B| \leq k + \sum_{i:=0,\ldots,d} |B_i| + |B_{>d}|
\leq k + d k^{d}+ dk=O(k^d),
$$
and the theorem is proved. \qed
\end{proof}

Next we consider \CapVC. In this problem, we are given a graph
$G$, an integer $k$, and a vertex capacity function $cap : V
(G) \to \N$, and the goal is to determine whether there exists
a vertex cover of size $k$ where each vertex covers no more
than its capacity. That is, whether there is a vertex cover $C$
of size $k$ and an injective mapping $\alpha$ mapping each edge
of $E(G)$ to one of its endpoints such that $|\alpha^{-1}(v)|
\leq cap(v)$ for every $v \in V(G)$. We may assume that $k+1 >
d$, since otherwise a kernel is trivially obtained by solving
the problem in polynomial time.

\begin{reduction-rule}
\label{Rule: CapVC}%
If $S \subseteq V(G)$ is a subset of twin vertices with a
common neighborhood $N(S)$ such that $|S| = k + 2 > d \geq
|N(S)|$, remove a vertex with minimum capacity in $S$ from~$G$,
and decrease all the capacities of vertices in $N(S)$ by one.
\end{reduction-rule}

\begin{lemma}
\label{Lemma: CapVC}%
Let $k \geq 1$ be an arbitrary integer, let $G$ be a vertex
capacitated graph, and let $G'$ be a vertex capacitated graph
resulting from applying either Rule~\ref{Rule: Isolated} or
Rule~\ref{Rule: CapVC} to $G$. Then $G$ has a capacitated
vertex cover of size $k$ iff $G'$ has a capacitated vertex
cover of size $k$.
\end{lemma}

\begin{proof}
The lemma is trivial for Rule~\ref{Rule: Isolated}. For
Rule~\ref{Rule: CapVC}, let $A$ be a capacitated vertex cover
of size $k$ in $G$, and let $u$ be a vertex of minimum capacity
in $S$. As $|S| > k$, there is some $v \in S \setminus A$, and
moreover it must be that $N(S) \subseteq A$. Thus, if $u \notin
A$, then $A$ is also a capacitated vertex cover of $G'$.
Otherwise, if $u \in A$, we can replace $u$ with $v$ in $A$. As
$cap(u) \leq cap(v)$, this would result in another capacitated
vertex cover for $G$ which is also a solution for $G'$.
Conversely, any capacitated vertex cover of size $k$ for $G'$
must also include all vertices of $N(S)$, and thus by
increasing the capacities of all these vertices by one we
obtain a solution of size $k$ for $G$. \qed
\end{proof}

\begin{theorem}
\CapVC\ has a kernel of size $O(k^{d+1})$ in $d$-degenerate
graphs.
\end{theorem}

\begin{proof}[sketch]
The argument is similar to the one used in
Theorem~\ref{Theorem: ConVC}. The only difference is that now
the size of sets $B_i$ is bounded by $(k+1)\binom{k}{i}$
instead of $d\binom{k}{i}$, which yields a kernel size of
$O(k^{d+1})$ instead of $O(k^d)$. \qed
\end{proof}

The following complimenting lower bound follows from a simple
reduction \dSC. In this problem we are given a $d$-regular
hypergraph $(V,\cE)$ and an integer $k$ and the question is
whether there are $E_1,\ldots,E_k \in \cE$ with $V =\cup_i
E_i$. Clearly we can assume that $|V| \leq dk$ since otherwise
there is no solution, and the problem remains hard also when
$|V|=dk$ (and the solution is a partition of $V$). Dell and
Marx show that unless coNP $\subseteq$ NP/poly, \dSC has no
kernel of size $O(k^{d-\varepsilon})$ for any $\varepsilon
>0$~\cite{DellMarx2012}. Note that this lower bound holds even
if the output of the kernel is an instance of another decidable
problem; in particular, even if the output is an instance of
\CapVC.

\begin{theorem}
Let $d \geq 3$. Unless \textnormal{coNP} $\subseteq$
\textnormal{NP/poly}, \CapVC in $d$-degenerate graphs has no
kernel of size $O(k^{d-\varepsilon})$ for any $\varepsilon >
0$.
\end{theorem}

\begin{proof}
Given an instance $(V,\cE,k)$ of \dSC with $|V|=kd$, we
construct a graph $H:=(U,F)$ by initially taking the incidence
bipartite graph on $V \cup \cE$, and then connecting each $v
\in V$ to new a leaf-vertex $v'$ which is adjacent only to $v$.
In this way, $U:=V' \cup V \cup \cE$, where $V':=\{v' : v \in
V\}$ is a set of copies of $V$, and $F:=\{\{v,E\}: v \in V, e
\in \cE, \text{ and } v \in e\} \cup \{\{v,v'\} : v \in V\}$.
To complete our construction, we set the capacity of each
vertex $u \in U$ to be its degree in $H$ minus one. Note that
$H$ is $d$-degenerate.

It is not difficult to see that $(V,\cE,k)$ has a solution iff
$H$ has a capacitated vertex cover of size $k+|V|=(d+1)k$.
Indeed if $E_1,\ldots,E_k$ is a solution for $(V,\cE,k)$, then
$\{E_1,\ldots,E_k\} \cup V$ is a capacitated vertex cover of
$H$. Conversely, any minimal capacitated vertex cover of $H$
must include all vertices of $V$ and none of $V'$, and hence if
it is of size $|V|+k$, it must include $k$ vertices which
correspond to $k$ edges of $\cE$ that cover $V$. Thus,
combining the above construction with a $O(k^{d-\varepsilon})$
kernel for \CapVC in $d$-degenerate graphs shows that coNP
$\subseteq$ NP/poly according to~\cite{DellMarx2012}. \qed
\end{proof}

\bibliographystyle{plain}
\bibliography{biblio}

\begin{thebibliography}{10}

\bibitem{AlberFellowsNiedermeier2004}
Jochen Alber, Michael~R. Fellows, and Rolf Niedermeier.
\newblock Polynomial-time data reduction for dominating set.
\newblock {\em Journal of the {ACM}}, 51(3):363--384, 2004.

\bibitem{Bodlaender-et-al2009a}
Hans~L. Bodlaender, Rodney~G. Downey, Michael~R. Fellows, and Danny Hermelin.
\newblock On problems without polynomial kernels.
\newblock {\em Journal of Computer and System Sciences}, 75(8):423--434, 2009.

\bibitem{Bodlaender-et-al2009c}
Hans~L. Bodlaender, Fedor~V. Fomin, Daniel Lokshtanov, Eelko Penninkx, Saket
  Saurabh, and Dimitrios~M. Thilikos.
\newblock ({M}eta) kernelization.
\newblock In {\em Proc. of the 50th annual {IEEE} symposium on Foundations Of
  Computer Science ({FOCS})}, pages 629--638, 2009.

\bibitem{Bodlaender-et-al2011}
Hans~L. Bodlaender, Bart M.~P. Jansen, and Stefan Kratsch.
\newblock Cross-composition: A new technique for kernelization lower bounds.
\newblock In {\em Proc. of the 28th international Symposium on Theoretical
  Aspects of Computer Science ({STACS})}, pages 165--176, 2011.

\bibitem{Bodlaender-et-al2009b}
Hans~L. Bodlaender, St{\'e}phan Thomass{\'e}, and Anders Yeo.
\newblock Kernel bounds for disjoint cycles and disjoint paths.
\newblock In {\em Proc. of the 17th annual European Symposium on Algorithms
  ({ESA})}, pages 635--646, 2009.

\bibitem{CaiChenDowneyFellows1997}
Liming Cai, Jianer Chen, Rodney~G. Downey, and Michael~R. Fellows.
\newblock Advice classes of paramterized tractability.
\newblock {\em Annals of Pure and Applied Logic}, 84(1):119--138, 1997.

\bibitem{ChenFlumMuller2011}
Yijia Chen, J{\"o}rg Flum, and Moritz M{\"u}ller.
\newblock Lower bounds for kernelizations and other preprocessing procedures.
\newblock {\em Theory of Computing Systems}, 48(4):803--839, 2011.

\bibitem{Cygan-et-al2012}
Marek Cygan, Marcin Pilipczuk, Michal Pilipczuk, and Jakub~Onufry Wojtaszczyk.
\newblock Kernelization hardness of connectivity problems in d-degenerate
  graphs.
\newblock {\em Discrete Applied Mathematics}, 160(15):2131--2141, 2012.

\bibitem{DellMarx2012}
Holger Dell and D{\'a}niel Marx.
\newblock Kernelization of packing problems.
\newblock In {\em Proc. of the 23rd annual ACM-SIAM Symposium On Discrete
  Algorithms {(SODA)}}, pages 68--81, 2012.

\bibitem{DellMelkebeek2010}
Holger Dell and Dieter van Melkebeek.
\newblock Satisfiability allows no nontrivial sparsification unless the
  polynomial-time hierarchy collapses.
\newblock In {\em Proc. of the 42th annual ACM Symposium on Theory Of Computing
  ({STOC})}, pages 251--260, 2010.

\bibitem{Diestel2005}
Reinhard Diestel.
\newblock {\em Graph Theory}.
\newblock Springer-Verlag, 3rd edition, 2005.

\bibitem{DomLokshtanovSaurabh2009}
Michael Dom, Daniel Lokshtanov, and Saket Saurabh.
\newblock Incompressibility through colors and {ID}s.
\newblock In {\em Proc. of the 36th International Colloquium on Automata,
  Languages and Programming ({ICALP})}, pages 378--389, 2009.

\bibitem{DowneyFellows1999}
Rodney~G. Downey and Michael~R. Fellows.
\newblock {\em Parameterized Complexity}.
\newblock Springer-Verlag, 1999.

\bibitem{Erman-et-al2010}
Rok Erman, Lukasz Kowalik, Matjaz Krnc, and Tomasz Walen.
\newblock Improved induced matchings in sparse graphs.
\newblock {\em Discrete Applied Mathematics}, 158(18):1994--2003, 2010.

\bibitem{Fellows2003}
Michael~R. Fellows.
\newblock Blow-ups, win/win's, and crown rules: {S}ome new directions in {FPT}.
\newblock In {\em Proc. of the 29th International Workshop on Graph theoretic
  concepts in computer science {(WG)}}, pages 1--12, 2003.

\bibitem{Fellows2006}
Michael~R. Fellows.
\newblock The lost continent of polynomial time: {P}reprocessing and
  kernelization.
\newblock In {\em Proc. of the 2nd International Workshop on Parameterized and
  Exact Computation {(IWPEC)}}, pages 276--277, 2006.

\bibitem{Fellows-et-al2009}
Michael~R. Fellows, Danny Hermelin, Frances~A. Rosamond, and St{\'e}phane
  Vialette.
\newblock On the parameterized complexity of multiple-interval graph problems.
\newblock {\em Theoretical Computer Science}, 410(1):53--61, 2009.

\bibitem{Fomin-et-al2010}
Fedor~V. Fomin, Daniel Lokshtanov, Saket Saurabh, and Dimitrios~M. Thilikos.
\newblock Bidimensionality and kernels.
\newblock In {\em Proc. of the 21st annual ACM-SIAM Symposium On Discrete
  Algorithms {(SODA)}}, pages 503--510, 2010.

\bibitem{FortnowSanthanam2008}
Lance Fortnow and Raul Santhanam.
\newblock Infeasibility of instance compression and succinct {PCP}s for {NP}.
\newblock In {\em Proc. of the 40th annual ACM Symposium on Theory Of Computing
  ({STOC})}, pages 133--142, 2008.

\bibitem{GuoNiedermeier2007b}
Jiong Guo and Rolf Niedermeier.
\newblock Invitation to data reduction and problem kernelization.
\newblock {\em SIGACT News}, 38(1):31--45, 2007.

\bibitem{GuoNiedermeier2007a}
Jiong Guo and Rolf Niedermeier.
\newblock Linear problem kernels for np-hard problems on planar graphs.
\newblock In {\em Proc. of 34th International Colloquium on Automata, Languages
  and Programming {(ICALP)}}, pages 375--386, 2007.

\bibitem{HarnikNaor2010}
Danny Harnik and Moni Naor.
\newblock On the compressibility of {NP} instances and cryptographic
  applications.
\newblock {\em SIAM Journal on Computing}, 39(5):1667--1713, 2010.

\bibitem{HermelinWu2012}
Danny Hermelin and Xi~Wu.
\newblock Weak compositions and their applications to polynomial lower-bounds
  for kernelization.
\newblock In {\em Proc. of the 23rd annual ACM-SIAM Symposium On Discrete
  Algorithms {(SODA)}}, pages 104--113, 2012.

\bibitem{Kanj-et-al2011}
Iyad~A. Kanj, Michael~J. Pelsmajer, Marcus Schaefer, and Ge~Xia.
\newblock On the induced matching problem.
\newblock {\em Journal of Computer and System Sciences}, 77(6):1058--1070,
  2011.

\bibitem{karp-21}
Richard~M. Karp.
\newblock Reducibility among combinatorial problems.
\newblock In R.~E. Miller and J.~W. Thatcher, editors, {\em Complexity of
  Computer Computations}, pages 85--103. Plenum Press, 1972.

\bibitem{Langer-et-al2012}
Eun~Jung Kim, Alexander Langer, Christophe Paul, Felix Reidl, Peter Rossmanith,
  Ignasi Sau, and Somnath Sikdar.
\newblock Linear kernels and single-exponential algorithms via protrusion
  decompositions.
\newblock In {\em Proc. of the 40th International Colloquium on Automata,
  Languages and Programming ({ICALP})}, page to appear, 2013.

\bibitem{Kratsch2012}
Stefan Kratsch.
\newblock Co-nondeterminism in compositions: a kernelization lower bound for a
  ramsey-type problem.
\newblock In {\em Proc. of the 23rd annual ACM-SIAM Symposium On Discrete
  Algorithms {(SODA)}}, pages 114--122, 2012.

\bibitem{NemhauserTrotter1975}
George~L. Nemhauser and Leslie E.~Trotter Jr.
\newblock Vertex packings: Structural properties and algorithms.
\newblock {\em Mathematical Programming}, 8(2):232--248, 1975.

\bibitem{Philip-et-al2012}
Geevarghese Philip, Venkatesh Raman, and Somnath Sikdar.
\newblock Polynomial kernels for dominating set in graphs of bounded degeneracy
  and beyond.
\newblock {\em ACM Transactions on Algorithms}, 9(1):11, 2012.

\end{thebibliography}

\end{document}